\documentclass[sigconf, nonacm]{acmart}

\usepackage[mathscr]{euscript}
\usepackage{xspace}
\usepackage[
	lambda,
	landau,
	operators,
	probability,
	sets,
	logic,
	complexity,
	asymptotics
]{cryptocode}
\usepackage{thmtools, thm-restate}
\usepackage{xcolor}
\usepackage{enumitem}
\usepackage{array}
\usepackage{booktabs}

\newcommand{\domain}{\mathfrak{D}}
\newcommand\sampledfrom{\mathrel{{\leftarrow}\vcenter{\hbox{\tiny\rmfamily\upshape\$}}}}

\definecolor{Gray}{RGB}{50,50,50}
\newcommand{\pcmycomment}[1]{\textcolor{Gray}{\texttt{//~#1}}}

\mathchardef\mhyphen="2D 

\newtheorem{definition}{Definition}[section]
\newtheorem{game}{Game}[section]
 
\newtheorem{theorem}{Theorem}[section]

\newtheorem{axiom}{Axiom}[]

\newcommand\vldbdoi{XX.XX/XXX.XX}
\newcommand\vldbpages{XXX-XXX}
\newcommand\vldbvolume{14}
\newcommand\vldbissue{1}
\newcommand\vldbyear{2020}
\newcommand\vldbauthors{\authors}
\newcommand\vldbtitle{\shorttitle} 
\newcommand\vldbavailabilityurl{https://github.com/jadidbourbaki/db_attacks}
\newcommand\vldbpagestyle{plain} 

\usepackage{comment}

\newif\ifshowcomments
\showcommentstrue 

  {\ifshowcomments\egroup\fi}

\begin{document}
\title{LSM Trees in Adversarial Environments}

\author{Hayder Tirmazi}
\affiliation{%
  \institution{City College of New York}
}
\email{hayder.research@gmail.com}






\begin{abstract}
The Log Structured Merge (LSM) Tree is a popular choice for key-value stores focusing on optimized write throughput while maintaining performant, production-ready read latencies. LSM stores rely on a probabilistic data structure called the Bloom Filter (BF) to optimize read performance. In this paper, we focus on adversarial workloads that lead to a sharp degradation in read performance by impacting the accuracy of BFs used within the LSM store. Our evaluation shows up to $800\%$ increase in the read latency of lookups for popular LSM stores. We define adversarial models and security definitions for LSM stores. We implement adversary resilience into two popular LSM stores, LevelDB and RocksDB. We use our implementations to demonstrate how performance degradation under adversarial workloads can be mitigated. 
\end{abstract}

\maketitle

\pagestyle{\vldbpagestyle}
\begingroup\small\noindent\raggedright\textbf{PVLDB Reference Format:}\\
\vldbauthors. \vldbtitle. PVLDB, \vldbvolume(\vldbissue): \vldbpages, \vldbyear.\\
\href{https://doi.org/\vldbdoi}{doi:\vldbdoi}
\endgroup
\begingroup
\renewcommand\thefootnote{}\footnote{\noindent
This work is licensed under the Creative Commons BY-NC-ND 4.0 International License. Visit \url{https://creativecommons.org/licenses/by-nc-nd/4.0/} to view a copy of this license. For any use beyond those covered by this license, obtain permission by emailing \href{mailto:info@vldb.org}{info@vldb.org}. Copyright is held by the owner/author(s). Publication rights licensed to the VLDB Endowment. \\
\raggedright Proceedings of the VLDB Endowment, Vol. \vldbvolume, No. \vldbissue\ %
ISSN 2150-8097. \\
\href{https://doi.org/\vldbdoi}{doi:\vldbdoi} \\
}\addtocounter{footnote}{-1}\endgroup

\ifdefempty{\vldbavailabilityurl}{}{
\vspace{.3cm}
\begingroup\small\noindent\raggedright\textbf{PVLDB Artifact Availability:}\\
The source code, data, and/or other artifacts have been made available at \url{\vldbavailabilityurl}
\endgroup
}

\section{Introduction}

A large number of modern key-value stores are based on Log-structured Merge Trees (LSM Trees)~\cite{shubham_subhadeep_2024,monkey}. LSM trees allow for write-optimized~\cite{monkey}, highly configurable~\cite{thakkar_2024} storage while being relatively simple to implement. These qualities have made them common as the storage engine for a large number of commercial and open-source database systems including LevelDB~\cite{level_db} from Google, bLSM~\cite{bLSM} and cLSM~\cite{cLSM} from Yahoo, RocksDB~\cite{rocks_db} and Cassandra~\cite{cassandra} from Meta, WiredTiger~\cite{wired_tiger} from MongoDB, Monkey~\cite{monkey} from Harvard's Data Systems Lab, Apache HBase~\cite{hbase}, ~\cite{badgerdb} from DGraph and many others~\cite{shubham_subhadeep_2024,wacky}.

LSM stores use Bloom Filters~\cite{bloom_1970,mitzenmacher_broder} to reduce unnecessary disk access. This strategy depends on Bloom Filters maintaining a low false positive rate (FPR). In normal workloads, this works well as Bloom Filters filter non-existent keys, so the LSM store does not have to look on disk~\cite{monkey}. An adversary can strategically insert keys into an LSM store that saturate its Bloom Filters, drastically raising their FPR. In this case, lookups for non-existing keys (called zero-result lookups~\cite{monkey}) require multiple disk accesses, increasing query latency by up to 800\% according to our experiments. An adversary only needs a modest number of well-chosen insertions to render Bloom Filters in LSM stores ineffective. This motivates the need for adversarially resilient mechanisms in LSM store.

\noindent We make the following contributions in this paper.
\begin{enumerate}
    \item We define rigorous adversarial models for LSM store and reason about their performance under computationally bound adversaries.
    \item We demonstrate significant performance degradation under adversarial workloads, showing an increase in lookup latency by up to 800\%. Our experiments are conducted on two popular LSM stores, LevelDB and RocksDB.
    \item We introduce a lightweight and provably secure mitigation strategy that employs keyed pseudorandom permutations (PRPs) to obfuscate key placement. 
    \item We implement our mitigation in LevelDB and RocksDB, and demonstrate that it reduces the impact of adversarial workloads while maintaining performance.
\end{enumerate}

\section{Preliminaries}\label{sec:preliminaries}

In this section, we first go over notation and rigorously define an LSM store. We then define a set of axioms that hold for a large number of popular LSM store implementations. Lastly, we use this axiomatic framework to introduce multiple aspects of LSM store design and performance include the performance of zero-result lookups and the use of auxiliary data structures.

\subsubsection*{Notation.} We review notation common in adversarial data structure literature~\cite{hayder_2025,filic_2024,filic_2022,naor_oved}. Given set $S$, we write $x \sampledfrom S$ to mean that $x$ is sampled uniformly randomly from $S$. For set $S$, we denote by $|S|$ the number of elements in $S$. The same notation is used for a list $\mathcal{L}$. We write variable assignments using $\leftarrow$. If the output is the value of a randomized algorithm, we use $\sampledfrom$ instead. For a randomized algorithm $\mathrm{A}$, we write $\text{output} \leftarrow \mathrm{A}_{r}(\text{input}_{1}, \text{input}_{2}, \cdots, \text{input}_{l})$, where $r \in \mathcal{R}$ are the random coins used by $\mathrm{A}$ and $\mathcal{R}$ is the set of possible coins. For natural number $n$, we denote the set $\{ 1, \cdots, n\}$ by $[n]$. Table~\ref{tab:notation} contains a summary of all the notation used in this paper.

\subsection{LSM stores} An LSM Tree consists of $\mathcal{N}$ levels~\cite{monkey}. $L_{i}$ indicates level $i$ in the LSM Tree. $L_{0}$ is typically an in-memory buffer, while the remaining levels $[L_{1}, \cdots, L_{\mathcal{N}}]$ consist of data in secondary storage including SSDs, Hard Disks, and external storage nodes connected over the network. We formalize the general syntax and behavior of an LSM store. Given an LSM store, $\Lambda$, we denote the set of public parameters of an LSM store by $\Phi$. The public parameters contain all the knobs relevant to the LSM store implementation including the number of levels in the LSM Tree as well as the knobs of any auxiliary data structures (discussed below) used by the LSM store. We denote the elements stored in $\Lambda$ by $\mathcal{L}_{\Lambda}$ (a list). We denote the state of $\Lambda$ by $\sigma \in \Sigma$ where $\Sigma$ is the space of all possible states of $\Lambda$. $\Lambda$ can store elements from any finite domain $\domain$, where $\domain = \cup_{l = 0}^{L}\{0, 1\}^{l}$ for any natural number $L \in \mathbf{N}$. An LSM store, $\Lambda$ consists of three algorithms.\vspace{0.2em}

\noindent \textbf{Construction}. $\sigma \sampledfrom C_{r}(\Phi)$ sets up the initial state of an empty LSM store with public parameters $\Phi$.\vspace{0.2em}

\noindent \textbf{Insertion} $\sigma^{\prime} \sampledfrom {I}_{r}((k, v), \sigma)$, given a tuple $(k, v) \in \domain \times (\domain \cup \{\bot\})$, returns the state $\sigma^{\prime}$ after insertion. After insertion, the key $k$ is an element of $\mathcal{L}_{\Lambda}$. The value $\bot$, referred to as a tombstone~\cite{monkey}, is used to indicate a deletion. Deletions have the same control flow as insertions in LSM stores~\cite{monkey}. \vspace{0.2em}

\noindent \textbf{Query}. $v \leftarrow Q(k, \sigma)$, given an element $x \in \domain$ returns a value $v \in \domain \cup \{ \bot \}$. The return value is either the value inserted by $I_{r}$ for key $k$, or $\bot$ if no such key was inserted. In the case of a deletion, a $\bot$ is returned as the tuple $(k, \bot)$ was inserted by $I_{r}$.\vspace{0.2em}

\noindent Note that, unlike $C_{r}$ and $I_{r}$, the query algorithm $Q$ is not allowed to change the value of the state. While $C_{r}$ and $I_{r}$ are allowed to use random coins, $Q$ is deterministic. A class of LSM stores can be uniquely identified by its algorithms: $\Lambda = (C_{r}, I_{r}, Q)$. We also assume all three algorithms always succeed and their outputs are correct with probability $1$.

\begin{table}
    \centering
    \renewcommand{\arraystretch}{1.2}
    \setlength{\tabcolsep}{4pt}
    \small
    \begin{tabular}{|c|l|}
        \hline
        \textbf{Symbol} & \textbf{Description} \\ 
        \hline
        \multicolumn{2}{|c|}{\textbf{General Notation}} \\ 
        \hline
        $\domain$ & Finite domain of elements \\ 
        \hline
        $S$ & A set of elements \\ 
        \hline
        $x \sampledfrom S$ & $x$ sampled randomly from $S$ \\ 
        \hline
        $|S|$ & Cardinality of set $S$ \\ 
        \hline
        $[n]$ & Set $\{1, \dots, n\}$ \\ 
        \hline
        \multicolumn{2}{|c|}{\textbf{LSM Store Notation}} \\ 
        \hline
        $\sigma \in \Sigma$ & State of LSM store \\ 
        \hline
        $\Lambda$ & An LSM store \\ 
        \hline
        $\Phi$ & Public parameters \\ 
        \hline
        $L_i$ & Level $i$ in LSM Tree \\ 
        \hline
        $R_{i,x}$ & Run $x$ in level $L_i$ \\ 
        \hline
        $N_{R}(L_i)$ & Number of runs in $L_i$ \\ 
        \hline
        $E_{C}(L_i)$ & Expected access cost in $L_i$ \\ 
        \hline
        $(k,v)$ & Key-value pair \\ 
        \hline
        $C_{r}(\Phi)$ & Construction algorithm\\
        \hline
        $Q(k, \sigma)$ & Query operation for key $k$ \\ 
        \hline
        $I_r((k, v), \sigma)$ & Insert operation \\ 
        \hline
        $\perp$ & Tombstone (deletion marker) \\ 
        \hline
        \multicolumn{2}{|c|}{\textbf{Bloom Filter Notation}} \\ 
        \hline
        $\Pi = ({C_{\dagger}}_r, Q_{\dagger})$ & Bloom Filter (BF) construction/query \\ 
        \hline
        $\sigma_{\dagger}$ & Internal state of BF \\ 
        \hline
        $m_{\dagger}$ & BF bit array size \\ 
        \hline
        $k_{\dagger}$ & Number of BF hash functions \\ 
        \hline
        $h_i(k)$ & $i^{th}$ BF hash function \\ 
        \hline
        $P_B(R_{i,x})$ & BF false positive probability \\ 
        \hline
        \multicolumn{2}{|c|}{\textbf{Adversarial Model Notation}} \\ 
        \hline
        $\mathcal{A} \in \mathbb{A}$ & Computationally-bound adversary \\
        \hline
        $\mathcal{A} \in \mathbb{A}_{BF}$ & BF-targeting adversary \\ 
        \hline
        $\lambda$ & Security parameter \\
        \hline
        $t$ & Number of adversary queries \\ 
        \hline
        \textsc{Smash-Lsm} & Adversarial game for LSM stores \\ 
        \hline
        \textsc{Smash-Bloom} & Adversarial game for Bloom Filters \\ 
        \hline
        $F_\kappa$ & Keyed pseudorandom permutation \\ 
        \hline
        $\epsilon$ & False positive/security bound \\ 
        \hline
        \multicolumn{2}{|c|}{\textbf{Oracle Notation}} \\
        \hline
        ${\mathscr{O}_{Q}(k)}$ & returns $\top$ if $\exists$ a BF $\Pi$ in LSM $\Lambda$ s.t $\Pi(k) = \top$\\
        \hline
        $\mathscr{O}_R$ & Returns state $\sigma_{\dagger}$ of each BF $\Pi_{i}$ in LSM $\Lambda$\\
        \hline
        $\mathscr{O}_I(k,v)$ & Inserts $(k,v)$ into LSM $\Lambda$\\
        \hline
        $\mathscr{O}_C(\Phi)$ & Constructs LSM $\Lambda$ with parameters $\Phi$ \\ 
        \hline
        ${\mathscr{O}_{\dagger}}_Q(k)$ & Queries BF $\Pi$ for $k$\\ 
        \hline
        ${\mathscr{O}_\dagger}_R$ & Returns state $\sigma_{\dagger}$ of BF $\Pi$ \\ 
        \hline
    \end{tabular}
    \vspace{0.2em}
    \caption{Mathematical notation used in the paper}
    \label{tab:notation}
\end{table}

\subsection{LSM Axioms} 

We define two axioms to reason about the performance of an LSM store, the Axiom of Cost (Axiom~\ref{axiom:access}) and the Axiom of Recency (Axiom~\ref{axiom:recency}). Our axioms hold for many popular LSM store implementations including LevelDB~\cite{level_db}, RocksDB~\cite{rocks_db}, Monkey~\cite{monkey}, WiredTiger~\cite{wired_tiger}, HBase~\cite{hbase} and Cassandra~\cite{cassandra}.

\begin{axiom}[Axiom of Cost]\label{axiom:access}
If $i > j$, then $E_{C}(L_{i}) \geq E_{C}(L_{j})$ where $E_{C}(L_{x})$ is the expected cost of accessing an entry in $L_{x}$.
\end{axiom}

\noindent LSM stores optimize for writes (inserts, deletes, and updates) using Axiom~\ref{axiom:access}. The LSM store immediately stores written entries in the in-memory buffer in $L_{0}$ without accessing slower storage in higher levels~\cite{monkey}. When the $L_{0}$ buffer reaches capacity, it is sorted (by key) and flushed to $L_{1}$. These sorted arrays are called \textit{runs}~\cite{monkey,shubham_subhadeep_2024,spooky}. We refer to a run $x$ within a level $L_{i}$ as $R_{i, x}$. We denote the number of runs within a level $L_{i}$ by $\mathcal{N}_{R}(L_{i})$. The value of $\mathcal{N}_{R}(L_{i})$ for a given LSM Tree depends on the merge policy of the LSM store~\cite{monkey}.

\begin{axiom}[Axiom of Recency]\label{axiom:recency}
For an entry $(k, v)$ in $R_{i, x}$ and an entry $(k, v^{\prime})$ in $R_{j, y}$ (with the same key $k$) if one of the following two conditions holds, then value $v$ was written before value $v^{\prime}$: $i > j$ (Case 1), or $i = j$ and $x > y$ (Case 2).
\end{axiom}

\noindent Using these axioms, we can explain LSM store design and performance.

\subsection{Zero-result Lookups.} 

In an LSM Tree with $L_{\mathcal{N}}$ levels, when any key $k$ is queried, the LSM store begins at the in-memory buffer at $L_{0}$ and traverses from $L_{0}$ to $L_{\mathcal{N}}$ in ascending order. If the LSM store finds a matching key at level $L_{i}$, the query returns early~\cite{monkey}. Entries at higher levels are superseded by more recent entries at lower levels thanks to Axiom~\ref{axiom:recency}, so looking further is unnecessary. For a key whose most recent entry is present in run $R_{i, x}$, the number of runs an LSM store needs to probe is at most $\sum_{l = 0}^{i} \sum_{r = 0}^{N_{R}(L_{l})} E_{C}(R_{l, r})$ where $E_{C}(r)$ is the expected I/O cost of probing run $r$.

A query on a key $k$ that is not stored in the LSM Tree is called a zero-result lookup. Such queries have a high worst-case I/O cost because the LSM store must probe every run within every level before the LSM store can be sure the key does not exist. Zero-result lookups are very common in practice~\cite{bLSM, monkey}. They have been the focus of LSM store analysis and new LSM store designs proposed by prior work~\cite{wacky,monkey}. A zero-result lookup is the worst-case lookup time~\cite{monkey} because the number of runs the LSM store must prove now is $\sum_{l = 0}^{\mathcal{N}} \sum_{r - 0}^{N_{R}(L_{l})} R_{l, r}$. This is the worst case for a point query because this is a probe of every run present in the LSM store.\vspace{0.2em}

\subsection{Auxiliary Data Structures} 

For each run $R_{i, x}$ in an LSM Tree, modern LSM stores store two auxiliary data structures in main memory~\cite{monkey}: a Bloom Filter (BF)~\cite{bloom_1970} and an array of fence pointers~\cite{huynh_et_al_2022,monkey}. The fence pointers contain min/max information for the disk pages that store run $R_{i, x}$~\cite{rocks_db,level_db}. In a point query, the LSM store uses this array of fence pointers when traversing a level to figure out which disk page to read when searching for a key~\cite{monkey}. The LSM store only has to read one disk page for each run. 

For a point query, the LSM store first probes a run's Bloom Filter. The LSM store only accesses the run in secondary storage if the corresponding BF returns positive. If the run indicated by the BF does indeed have the key, the LSM store returns the query early following Axiom~\ref{axiom:access}. However, the BF is a probabilistic data structure with one-sided errors, so it can have false positives with some probability~\cite{bloom_1970,mitzenmacher_broder}. In the case of a false positive, the query continues to run~\cite{monkey}. Having a Bloom Filter in main memory for each run reduces the expected I/O cost of the LSM store for a point query, depending on the false positive probability of the BF. For a key whose most recent entry is present in run $R_{i, x}$, the LSM store needs to probe, in expectation: \[ E_{C}(R_{0,0}) \cdot ( \sum_{l = 0}^{i} \sum_{r = 0}^{N_{R}(L_{i - 1})} P_{B}(R_{l, r}) E_{C}(R_{l, r}) + \sum_{r = 0}^{x - 1} P_{B}(R_{i, r}) R_{i, r} + R_{i, x}))\]

\noindent where $P_{B}(r)$ is the false positive probability of the BF corresponding to run $r$. The extra $E_{C}(R_{0,0})$ factor comes from the fact that each BF resides in the main memory, which involves the cost of accessing a main memory page. Since $R_{0, 0}$ (the run at $L_{0}$) is also a main memory page. Therefore, the expected cost of accessing a main memory page is the same as that of accessing run $R_{0,0}$.
 
\section{Adversarial Model}\label{sec:adversarial_model}

In this section, we first discuss the performance degradation caused by adversarial attacks and the feasibility of such attacks. We then introduce a game-based adversarial model for LSM stores and propose a security definition for LSM stores.

\begin{figure}
    \centering
    \includegraphics[width=0.65\linewidth]{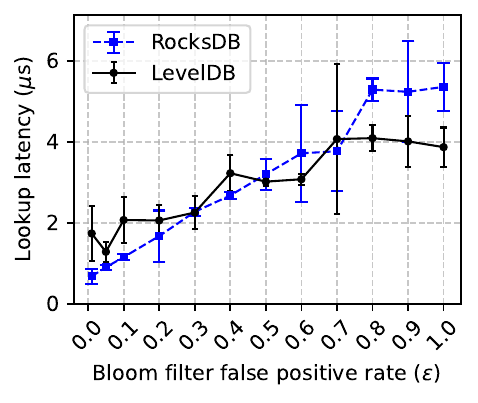}
    \caption{Performance degradation of zero-result lookups on a uniformly random query workload}
    \Description{Graph showing performance degradation of point queries for RocksDB on a uniformly random workload.}
    \label{fig:zero_result_uniform_random}
\end{figure}

\subsection{Performance Degradation} 
The key idea for our adversarial workloads is to target the false positive rate (FPR) of the per-run Bloom Filters implemented in an LSM store. Since most popular LSM stores~\cite{level_db,rocks_db,monkey} use BFs with non-cryptographic hashes, such attacks are feasible and inexpensive. Even if an adversary is allowed a very small bound on insertions, they can still fully insert enough elements to saturate the BFs. For a Standard Bloom Filter implementation with a memory budget of $m_{\dagger}$ bits and $k_{\dagger}$ hashes, only $\frac{m_{\dagger}}{k_{\dagger}}$ adversarial insertions are required to make the BF's FPR become $1$. Figure~\ref{fig:zero_result_uniform_random} shows the impact of Bloom Filter FPR on the zero-result lookup latency for LevelDB~\cite{level_db} and RocksDB~\cite{rocks_db}. We insert $10$M keys in batch sizes of $10$K. We benchmark $50$K uniformly randomly sampled zero-result look-ups for each experiment. Between the insertion stage and the lookup stage, we save and re-open the store to mitigate the effects of caching. As BF FPR increases, the lookup latency of LevelDB and RocksDB becomes $2$\texttt{x} and $8$\texttt{x} respectively.\vspace{0.2em}

\subsection{Feasibility of Attacks} 

We first give an overview of how the standard implementation of a Bloom Filter first suggested in~\cite{bloom_1970} and used in LevelDB~\cite{level_db} and RocksDB~\cite{rocks_db} works. A Standard Bloom Filter (SBF) construction is a zero-initialized array of $m_{\dagger}$ bits~\cite{mitzenmacher_broder,gerbet_2015} and requires a family of $k_{\dagger}$ independent hash functions, $\,\, h_{i, m}: \domain \mapsto [m_{\dagger}]$ for all $i \in [k_{\dagger}]$\cite{mitzenmacher_broder,hayder_2025}. Upon setup, For each element $x \in S$ ($S$ is the set being encoded by the Bloom Filter), the bits $h_{i}(x)$ are set to $1$ for $i \in [k_{\dagger}]$. When querying an element $x$, we return true if all $h_{i}(x)$ map to bits that are set to $1$. If there exists an
$h_{i}(x)$ that maps to a bit that is $0$, we return false.

\begin{figure}[H]
  \centering
\scalebox{1.0}{ 
\begin{minipage}{\columnwidth} 
    \centering
    \includegraphics[trim={4.5cm 6cm 4.5cm 2cm},width=0.32\textwidth,clip]{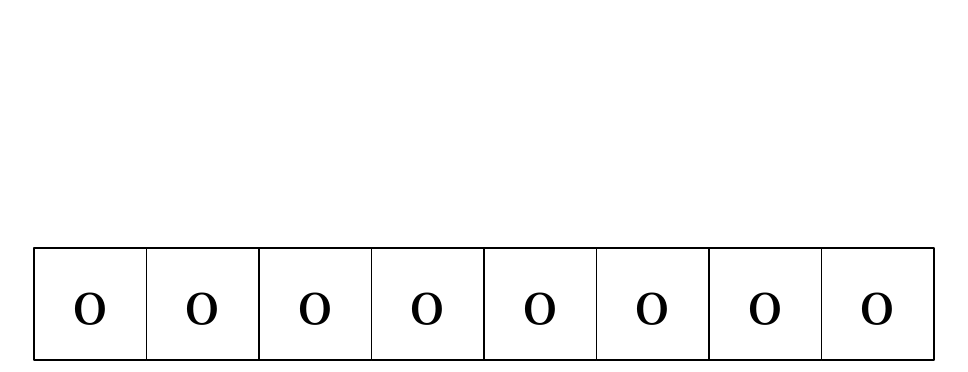}
    \includegraphics[trim={4.5cm 6cm 4.5cm 2cm},width=0.32\textwidth,clip]{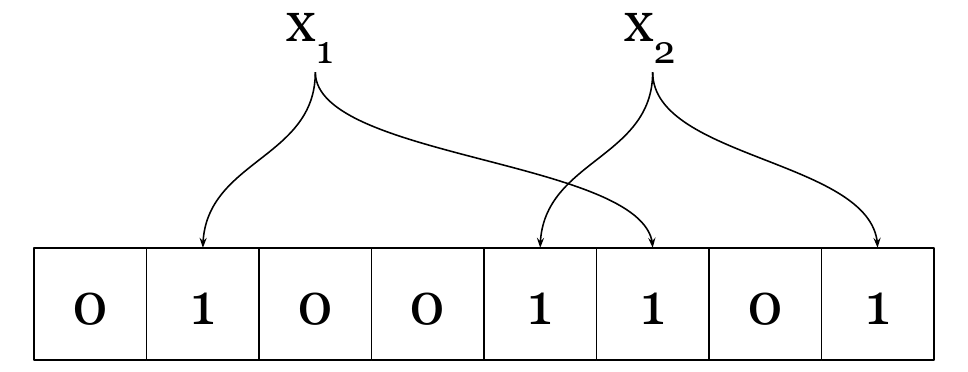}
    \includegraphics[trim={4.5cm 6cm 4.5cm 2cm},width=0.32\textwidth,clip]{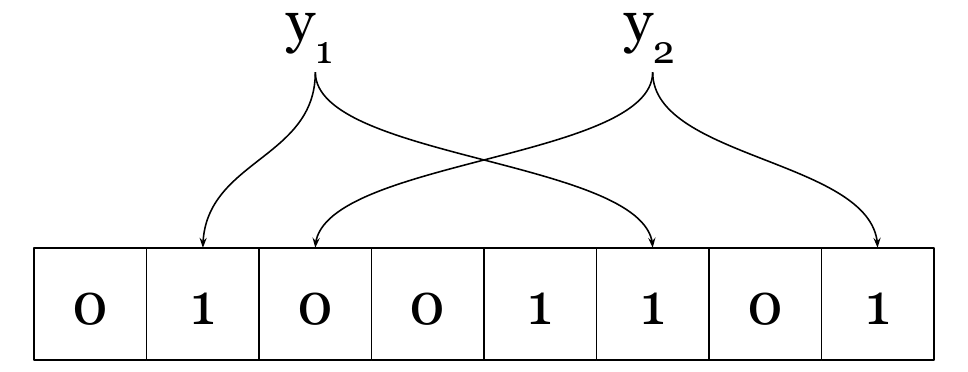}
\end{minipage}
}
    \caption{Bloom Filter example ($m_{\dagger}$ = 8, $k_{\dagger}$ = 2) adapted from~\cite{hayder_2024}. Inserted $x_{i}$ is hashed $k_{\dagger}$ times, setting mapped bits. Queried $y_{i}$ is hashed $k$ times. If a mapped bit is unset, $y_{i} \notin S$. Otherwise $y_{i}$ is in $S$ or a false positive}\label{fig:bloomintro}
    \Description{Three graphs showing how a Bloom Filter is initialized, performs inserts, and performs queries.}
\end{figure}

The expected number of entries to fully saturate an SBF is $\lfloor \frac{m_{\dagger}\log{m_{\dagger}}}{k_{\dagger}} \rfloor $~\cite{gerbet_2015}. An adversary can pick well-chosen items to insert such that sets $k_{\dagger}$ previously unset bits (one new bit for each hash function). This reduces the number of entries to fully saturate an SBF down to $\lfloor \frac{m_{\dagger}}{k_{\dagger}} \rfloor$. ~\cite{gerbet_2015} show many practical saturation attacks for real-world Standard Bloom Filter deployments. This is particularly easy to do when a non-cryptographic hash is invertible. However, even a brute force strategy where the adversary crafts well-chosen inputs offline does not take a large amount of time if the memory budget of a Bloom Filter is small. Since the entire point of a Bloom Filter is for it to use a small number of bits (otherwise we can simply replace it with a non-probabilistic data structure such as a hash-table-based set), this is commonly the case. Our inexpensive experimental setup (Section~\ref{sec:resilience}) running a sequential brute force algorithm can fully saturate LevelDB's Bloom Filter implementation with a memory budget of up to $4$ kilobits in about $11$ seconds.

\begin{figure}
    \centering
    \includegraphics[width=0.75\linewidth]{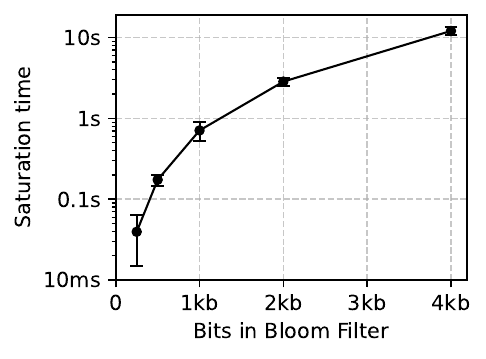}
    \caption{Time taken by a brute force algorithm running sequentially on a local machine to saturate LevelDB's Bloom Filter implementation with various memory budgets.}
    \Description{A line graph showing a brute force algorithm taking O(seconds) to saturate Bloom Filters}
    \label{fig:bloom_filter_brute_force_saturation}
\end{figure}

\begin{figure*}
    \centering
    \includegraphics[width=0.39\linewidth]{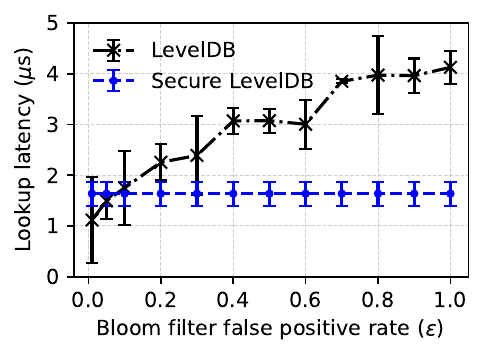}
    \includegraphics[width=0.38\linewidth]{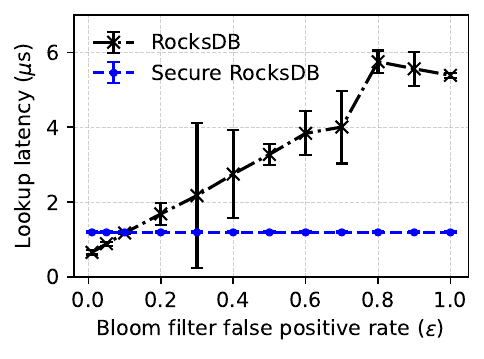}
    \caption{Zero-result lookup performance on a uniformly random query workload for LevelDB ( left) and RocksDB (right) with adversarial resilience.}
    \label{fig:zero_result_evaluation_secure}
\end{figure*}

\subsection{Game-based Model}

We define a self-contained game-based~\cite{naor_eylon,naor_oved,hayder_2024} adversarial model for LSM stores that is restricted to adversaries that target an LSM store's Bloom Filters. For cryptography-focused readers, we also define a more general simulator-based~\cite{lindell_2016} adversarial model for any class of adversaries in Appendix~\ref{sec:simulator_based_adversarial_model}. We will be assuming our adversary is \textit{efficient}. We will assume that the adversary works in non-uniform probabilistic polynomial time (n.u. p.p.t). This standard notion is used to model efficient adversaries in cryptography literature~\cite{pass_shelat}. Moreover, we will restrict our model to $\mathbb{A}_{\text{BF}}$, the set of efficient adversaries that target the Bloom Filters in the LSM store. We discuss other adversarial targets in Appendix~\ref{sec:adversarial_targets}. To capture the notion of adversary resilience in an LSM store setting, we propose a game inspired by prior work on game-based adversary models for other probabilistic data structures~\cite{naor_eylon, naor_oved}. We first define the notion of a Bloom Filter in our adversarial game. 

\begin{definition}[Bloom Filter]\label{def:bloom_filter} Let $\Pi = ({C_{\dagger}}_{r}, {Q}_{\dagger})$  be a pair of polynomial time algorithms. ${C_{\dagger}}_{r}$ is randomized, it takes a set $S_{\dagger} \subseteq \mathscr{D}$ as input and outputs a state $\sigma_{\dagger}$. $Q_{\dagger}$ is deterministic, it takes as a state $\sigma_{\dagger}$ and a key $k \in \mathscr{D}$ as input and returns $y \ in \{\top, \bot\}$. $\Pi$ is an $(n, \epsilon)$-BF if for all sets $S^{\Pi} \subseteq \domain$ of cardinality $n$ and suitable public parameters $\Phi_{\dagger}$, the following two properties hold.

\begin{itemize}
    \item Completeness: $\forall{x} \in S : P[Q_{\dagger}(x, {C_{\dagger}}_{r}(\Phi_{\dagger}, S_{\dagger})) = \top] = 1$\vspace{0.2em}
    \item Soundness: $\forall x \notin S : P[Q_{\dagger}(x, {C_{\dagger}}_{r}(\Phi_{\dagger}, S_{\dagger})) = \top] \leq \epsilon$
\end{itemize}

\noindent where the probabilities are over the coins of ${C_{\dagger}}_{r}$.
\end{definition}

\noindent In our game, the adversary $\mathcal{A} = (\mathcal{A}_{C}, \mathcal{A}_{Q})$ consists of two parts: $\mathcal{A}_{C}$ chooses a list $\mathcal{L} \subset \mathscr{D} \times (\mathscr{D} \cup \{ \bot \}$. $\mathcal{A}_{Q}$ gets $\mathcal{L}$ as input and attempts to find a false positive key $k$ given only oracle access to the LSM store $\Lambda$ initialized with $\mathcal{L}$. $\mathcal{A}$ succeeds if key $k$ is not among the queried elements and is a false positive. We measure the success probability of $\mathcal{A}$ for the random coins in $\Lambda$ and in $\mathcal{A}$. For computationally bound adversaries, our game includes a security parameter $\lambda$ which is given to the adversary $\mathcal{A}$ in unary $1^{\lambda}$ and given to the LSM store implicitly as part of the public parameters $\Phi$. For LSM store $\Lambda = (C_{r}, I_{r}, Q)$, the adversary $\mathcal{A}$ is allowed access to two oracles. The first oracle, ${\mathscr{O}_{Q}(k)}$, returns $\top$ if there exists a Bloom Filter $\Pi$ in LSM store $\Lambda$ such that $\Pi(k) = \top$. Otherwise, it returns $\bot$. The second oracle, ${\mathscr{O}_{R}}$ returns a list ${\sigma_{\dagger}}_{i}$ of internal states for each Bloom Filter $\Pi_{i}$ in the LSM store $\Lambda$.

\begin{game}[\textsc{Smash-Lsm}]\label{game:smash_lsm}
    We have a challenger $\Upsilon$, security parameter $\lambda$, and a n.u p.p.t adversary $\mathcal{A} = (\mathcal{A}_{C}, \mathcal{A}_{Q})$. $\mathcal{A}$ is a Bloom Filter targeting adversary, $\mathcal{A} \in \mathbb{A}_{\text{BF}}$. We have an LSM store $\Lambda$ with public parameters $\Phi$. We define the game \textsc{Smash-Lsm}($\mathcal{A}, t, \lambda$) as follows. 
    \begin{description}
        \item[Step 1] $\mathcal{L} \sampledfrom A_{C}(1^{\lambda})$ where $ \mathcal{L} \subset \mathscr{D} \times (\mathscr{D} \cup \{ \bot \})$ and $|\mathcal{L}| = n.$
        \item[Step 2]  $\Upsilon$ initializes $\Lambda$ with $\sigma \sampledfrom C_{r}(\Phi)$ and invokes $\sigma \sampledfrom I_{r}((k, v), \sigma)$ for each $(k, v) \in \mathcal{L}$.
        \item[Step 3] $k_{\mathcal{A}} \sampledfrom \mathcal{A}_{Q}(1^{\lambda}, \mathcal{L})$. $\mathcal{A}_{Q}$ can make at most $t$ queries $k_1, \cdots, k_{t}$ to $\mathscr{O}_{Q}$ and unbounded queries to $\mathscr{O}_{R}$.
        \item[Step 4] We denote the set of keys in $\mathcal{L}$ by $L_{k}$ . If $k_{\mathcal{A}} \notin \mathcal{L}_{k} \cup \{ k_{1}, \cdot, k_{t} \}$ and $\mathscr{O}_{Q}(k_{\mathcal{A}}) = \top$, $\mathcal{A}$ wins (\textsc{Smash-Lsm} returns $\top$). Otherwise, $\mathcal{A}$ loses (\textsc{Smash-Lsm} returns $\bot$). 
    \end{description}
\end{game}

\noindent Clearly, if adversary $\mathcal{A}$ wins the $\textsc{Smash-Lsm}$ game, $\mathcal{A}$ has \emph{forced an unnecessary run access} by exploiting Bloom Filter false positives in the LSM store.\vspace{0.2em}

\subsection{Security Definition} 

We now define our notion of the security of an LSM store $\Lambda$ against Bloom Filter targeting adversaries $\mathbb{A}_{\text{BF}}$ using the \textsc{Smash-LSM} game.

\begin{definition}\label{def:secure_lsm_tree}
    For an LSM store $\Lambda$, we say that $\Lambda$ is $(n, t, \varepsilon)$-secure against Bloom Filter targeting adversaries if for all n.u p.p.t adversaries $\mathcal{A} \in \mathbb{A}_{\text{BF}}$ and for all lists of cardinality $n$, for all large enough $\lambda \in \mathbb{N}$, it holds that
    \[
    \text{Pr}[\textsc{Smash-Lsm}(\mathcal{A}, t, \lambda) = \top] \leq \varepsilon
    \]

    \noindent where the probabilities are taken over the random coins of $\Lambda$ and $\mathcal{A}$.
\end{definition}

\noindent The definition essentially states a requirement for declaring that an LSM store is  secure against computationally bound adversaries that are looking to sabotage its Bloom Filters. We have to prove that no such adversary can, with high probability, find a previously unused key that results in a false positive for any of the Bloom Filters being used by the LSM store. This must be proven, even if the adversary chooses the initial list of keys to be inserted into the LSM store, even if the adversary makes $t$ queries to the LSM store's Bloom Filters before answering, and even if the adversary can look at the internal state of the Bloom Filters.

\section{Adversary Resilience}\label{sec:resilience}

We construct an adversary-resilient implementation of LevelDB~\cite{level_db} and RocksDB~\cite{rocks_db}. Our secure construction is simple and easily pluggable in any popular LSM store. Instead of writing and reading keys directly, we simply patch the LSM store API to maintain a keyed pseudorandom permutation (PRP) and read/write the permuted value of the key. The values are kept unchanged. This does not affect the correctness of the LSM store. As an example, assume we use a PRP that permutes $a$ to $x$ and $b$ to $y$. Inserting two key-value pairs $(a, v_{1})$ and $(b, v_{2})$ will instead insert $(x, v_{1})$ and $(y, v_{2})$. When key $a$ is queries, since our PRP is consistent, it is mapped again to $x$ and $\Lambda$ returns the correct corresponding value $v_{1}$. PRPs are bijections therefore the keys themselves can also be recovered by running the inverse permutation. We informally summarize our result here and then rigorously prove it below in Section~\ref{sec:security_proofs}.

\begin{theorem}
    Let $\Lambda = (C_{r}, I_{r}, Q)$ be an LSM store using $m$ bits of memory for its Bloom Filters. If pseudo-random permutations exist, then there exists a negligible function $\mathrm{negl(\cdot)}$ such that for security parameter $\lambda$, there exists an LSM engine $\Lambda^{\prime}$ that is $(n, t, \epsilon + \mathrm{negl}(\lambda))$-secure against Bloom Filter targeting adversaries $\mathbb{A}_{\text{BF}}$ using ${m}^{'} = m + \lambda$ bits of memory for its Bloom Filters.
\end{theorem}

\noindent \textit{Proof Sketch:} Applying a keyed PRP to LSM store keys before insertion and query ensures that adversaries cannot control key placements to force high false positive rates. Correctness is due to PRP bijectivity: queries map consistently to transformed keys, retrieving the correct values. Security follows from PRP indistinguishability. If an adversary could significantly increase false positives, they could distinguish the PRP from random, which contradicts the existence of PRPs. We need extra $\lambda$ bits of memory to store the PRP's secret key.\vspace{0.2em}

\begin{figure*}
    \centering
    \includegraphics[width=0.39\linewidth]{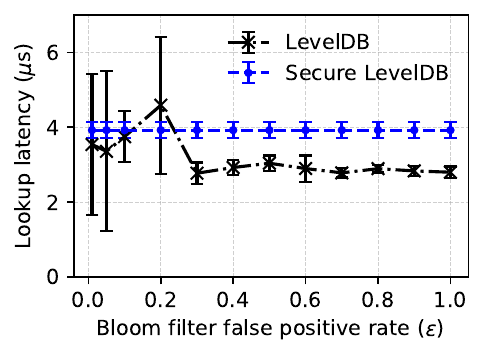}
    \includegraphics[width=0.39\linewidth]{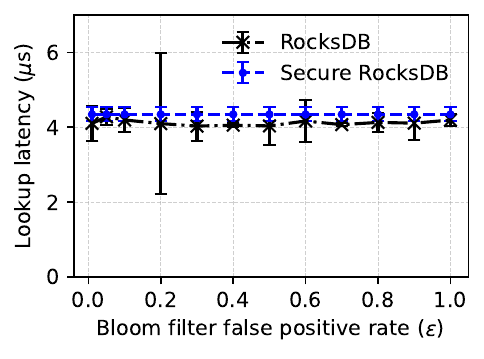}
    \caption{Existing query lookup performance on a uniformly random query workload for LevelDB (left) and RocksDB (right) with adversarial resilience.}
    \label{fig:existing_result_evaluation_secure}
\end{figure*}

\subsection{Implementation \& Experimental Setup} 

We use the experimental setup discussed here for all experiments in this work. We use an Apple M2 processor with $8$ cores, $8$ GB of Memory, a $128$ KB L1 cache, a $1$ MB L2 cache, and $256$ GB of SSD storage out of which $245$ GB is usable storage. We evaluated the most recent versions of RocksDB ($9.10$) and LevelDB ($1.23$). We run each experiment $5$ times and display the median, the error bars indicate standard deviations. We implement security as an easily pluggable module written in C++ 17 for both LevelDB and RocksDB. Our implementation relies on the hardness of AES. Our implementation uses the AES-128 implementation provided by OpenSSL v3.4.0 in CTR mode. We use a $16$ byte ($\lambda = 128$ bits) AES secret key to parameterize the PRP. To prevent secret key leakage, production-scale deployments can rely on secret stores and ephemeral key rotation mechanisms. We have released our implementation and evaluation as open-source software.

\subsection{Performance} 
We compare our adversary-resilient implementations to untouched releases of LevelDB and RocksDB under adversarial workloads. Figures\ref{fig:zero_result_evaluation_secure} and~\ref{fig:existing_result_evaluation_secure} show performance under the same workload as Section~\ref{sec:adversarial_model}. For zero-result lookups, our implementations reduce latency by 60\% and 78\% for LevelDB and RocksDB respectively. Our implementations are slower for workloads with existing keys (where not all the runs need to be probed, in expectation) by 30\% and 6\% for LevelDB and RocksDB respectively. We also measured the latency for $50$ K uniformly random inserts with our implementation and saw a median latency increase of 39\% and 49\% for LevelDB and RocksDB respectively. These increases are due to the extra compute overhead of AES for every key. This is a constant compute cost that can potentially be reduced by using a hardware-implemented version of AES such as AES-NI~\cite{naor_eylon} for Intel processors. This overhead also is not a major concern when I/O dominates lookup performance, which is frequently the case for LSM stores~\cite{monkey}.

\section{Security Proofs}\label{sec:security_proofs}

We first define the notion of secure Bloom Filters (Def.~\ref{def:bloom_filter}). We use known definitions~\cite{filic_2022,hayder_2025} for Bloom Filter (BF) constructions. A BF, $\Pi$ consists of two algorithms.\vspace{0.2em}

\noindent \textbf{Construction}. $\sigma_{\dagger} \leftarrow {C_{\dagger}}_{r}(\Phi_{\dagger}, S_{\dagger})$ sets up the initial state of a BF with public parameters $\Phi_{\dagger}$ and a given set $S_{\dagger} \subseteq \domain$.\vspace{0.2em}


\noindent \textbf{Query}. $b \leftarrow Q_{\dagger}(x, \sigma_{\dagger})$, given an element $x \in \domain$ returns a boolean $b \in \{ \bot, \top \}$. The return value approximately answers whether $x \in S_{\dagger}$ ($b = \top$) or $x \notin S_{\dagger}$ ($b \neq \bot$).\vspace{0.2em}

\noindent The construction algorithm $C_{r}$ is called first to initialize $\Pi$. The query algorithm $Q$ is not allowed to change the value of the state. While $C_{r}$ is randomized, $Q$ is deterministic. Both algorithms always succeed. A class of BFs can be uniquely identified by its algorithms: $\Pi = (C_{r}, Q)$.

We now define a well-known~\cite{naor_eylon,naor_oved} security game for Bloom Filters. Our adversary, $\mathcal{A}_{\dagger} = ({\mathcal{A}_{\dagger}}_{C}, {\mathcal{A}_{\dagger}}_{Q}$ consists of two parts: ${\mathcal{A}_{\dagger}}_{C}$ chooses a set $S_{\dagger} \subset \mathscr{D}$. ${\mathcal{A}_{\dagger}}_{Q}$ gets $S_{\dagger}$ as input and attempts to find a false positive key $k$ given only oracle access to the BF $\Pi$ initialized with $S_{\dagger}$. We measure the success probability of $\mathcal{A}_{\dagger}$ for the random coins in $\Pi$ and in $\mathcal{A}_{\dagger}$. For computationally bound adversaries, our game includes a security parameter $\lambda$ which is given to the adversary $\mathcal{A}_{\dagger}$ in unary $1^{\lambda}$ and given to the BF implicitly as part of the public parameters $\Phi_{\dagger}$.

For a BF $\Pi = ({C_{\dagger}}_{r}, Q_{\dagger})$, the adversary $\mathcal{A}_{\dagger}$ is allowed access to two oracles. The first oracle, ${\mathscr{O}_{Q_{\dagger}}(k)}$, returns ${C_{\dagger}}_{r}(k)$. The second oracle, ${\mathscr{O}_{R_{\dagger}}}$ returns ${\sigma_{\dagger}}$, the internal states for the Bloom Filter $\Pi$. Note that the internal state returned does \textbf{not} include any secret keys used in the construction. This assumption is consistent with prior work~\cite{naor_oved, naor_eylon, filic_2022,filic_2024,hayder_2025}.

\begin{game}[\textsc{Smash-Bloom}]\label{game:smash_bloom}
    We have a challenger $\Upsilon$, security parameter $\lambda$, and a n.u p.p.t adversary $\mathcal{A}_{\dagger} = ({\mathcal{A}_{\dagger}}_{C}, {\mathcal{A}_{\dagger}}_{Q}$. We have a Bloom Filter $\Pi$ with public parameters $\Phi_{\dagger}$. We define the game \textsc{Smash-Bloom}($\mathcal{A}_{\dagger}, t, \lambda$) as follows. 
    \begin{description}
        \item[Step 1] $S_{\dagger} \sampledfrom A_{{C_{\dagger}}}(1^{\lambda})$
        \item[Step 2]  $\Upsilon$ initializes $\Pi$ with $\sigma_{\dagger} \sampledfrom {C_{\dagger}}{r}(\Phi_{\dagger}, S_{\dagger})$.

        \item[Step 3] $k_{\mathcal{A}} \sampledfrom \mathcal{A}_{Q_{\dagger}}(1^{\lambda}, S_{\dagger})$. $\mathcal{A}_{Q_{\dagger}}$ can make at most $t$ queries $k_1, \cdots, k_{t}$ to $\mathscr{O}_{Q_{\dagger}}$ and unbounded queries to $\mathscr{O}_{R_{\dagger}}$.
        \item[Step 4] If $k_{\mathcal{A}} \notin S_{\dagger} \cup \{ k_{1}, \cdot, k_{t} \}$ and $\mathscr{O}_{Q_{\dagger}}(k_{\mathcal{A}}) = \top$, $\mathcal{A}$ wins (\textsc{Smash-Bloom} returns $\top$). Otherwise, $\mathcal{A}$ loses (\textsc{Smash-Bloom} returns $\bot$). 
    \end{description}
\end{game}

\begin{definition}\label{def:secure_bloom_filter}
    For an Bloom Filter $\Pi$, we say that $\Pi$ is $(n, t, \varepsilon)$-secure if for all n.u p.p.t adversaries $\mathcal{A}_{\dagger}$ and for all sets of cardinality $n$, for all large enough $\lambda \in \mathbb{N}$, it holds that
    \[
    \text{Pr}[\textsc{Smash-Bloom}(\mathcal{A}_{\dagger}, t, \lambda) = \top] \leq \varepsilon
    \]

    \noindent where probabilities are over the random coins of $\Pi$ and $\mathcal{A}_{\dagger}$.
\end{definition}

\noindent \textbf{Pseudo-random Permutations.} We provide a brief self-contained treatment of a cryptographic construction called pseudorandom permutations adapted from ~\cite{katz_lindell_2014} that will allow the construction of secure Bloom Filters. Let $\text{Perm}_{n}$ be the set of all permutations on $\{0, 1\}^{n}$.

\begin{definition}\label{def:efficientfunction}
 Let an efficient permutation $F$ be any permutation for which there exists a polynomial time algorithm to compute $F_{k}(x)$ given $k$ and $x$, and there
 also exists a polynomial time algorithm to compute $F_{k}^{-1}(x)$ given $k$ and $x$.
\end{definition}

\begin{definition}\label{def:keyedpermutation}
 Let $F : {\{0, 1\}}^{*} \times {\{0, 1\}}^{*} \mapsto {\{0, 1\}}^{*}$ be 
 an efficient, length-preserving, keyed function. $F$ is a keyed permutation if $\forall k$, $F_{k}(\cdot)$ is one-to-one.
\end{definition}

\begin{definition}\label{def:pseudo-randompermutation}
 Let $F: {\{0, 1\}}^{*} \times {\{0, 1\}}^{*} \mapsto {\{0, 1\}}^{*}$ be an efficient keyed permutation. $F$ is a pseudo-random 
 permutation if for all probabilistic polynomial time distinguishers $D$, there exists a negligible function $\mathtt{negl}$, such 
 that 
 \begin{align*}
 |\Pr[D^{F_{k}(\cdot)F_{k}^{-1}(\cdot)}(1^{n}) = 1] - \Pr[D^{f_{n}(\cdot)f_{n}^{-1}(\cdot)}(1^{n}) = 1]| \\ 
 \leq \mathtt{negl}(n)
 \end{align*}
 where the first probability is taken over uniform choice of $k \in \{0, 1\}^{n}$ and the randomness of $D$, and the second probability is taken over uniform choice of $f \in \text{Perm}_{n}$ and the randomness of $D$.
\end{definition}

\noindent \textbf{Secure Bloom Filters.} We show a well-known result for Bloom Filters initially proved by~\cite{naor_eylon}. The proof was expanded by~\cite{hayder_2024} to make it clearer that it holds even for the case where an adversary has access to the internal state of a Bloom Filter. We include a self-contained proof here using our notation based on proofs in the two cited works.

\begin{theorem}\label{thm:secure_bloom_filter}
    Let $\Pi = ({C_{\dagger}}_{r}, Q_{\dagger})$ be a Bloom Filter using $m_{\dagger}$ bits of memory. If pseudo-random permutations exist, then there exists a negligible function $\mathrm{negl(\cdot)}$ such that for security parameter $\lambda$, there exists an $(n, t, \epsilon + \mathrm{negl}(\lambda))$-secure BF using ${m_{\dagger}}^{'} = m_{\dagger} + \lambda$ bits of memory.
\end{theorem}

\begin{proof} We first demonstrate a secure construction. We then prove its security and correctness.\vspace{0.2em}

\noindent \textit{Construction}: Choose a key $\kappa \in \{0, 1\}^{\lambda}$ for a pseudorandom permutation $F_{\kappa}$. Let $\Pi^{\prime} = ({C_{\dagger}}^{\prime}_{r}, Q^{\prime}_{\dagger})$ where

\begin{enumerate}
    \item ${C_{\dagger}}^{\prime}_{r}(\Phi_{\dagger}, S_{\dagger}) = {C_{\dagger}}_{r}(\Phi_{\dagger}, 
S^{\prime}_{\dagger})$ where $S^{\prime}_{\dagger}$ is the permuted set $S_{\dagger}$ i.e. $S^{\prime}_{\dagger} = \{ F_{\kappa}(x): x \in S_{\dagger} \}$. 
    \item $Q^{\prime}_{\dagger}(x, \sigma_{\dagger}) = Q_{\dagger}(F_{\kappa}(x), \sigma_{\dagger})$.
\end{enumerate}

\noindent ${C^{\prime}_{\dagger}}_{r}$ initializes Bloom Filter $\Pi^{\prime}$ with $S^{\prime}_{\dagger}$. $Q^{\prime}_{\dagger}$ on input $x$ queries for $x^{'} = F_{\kappa}(x)$. The only additional memory required is for storing $\kappa$ which is $\lambda$ bits long.

\textit{Security Proof:} The security of $\Pi^{\prime}$ follows from a hybrid argument. Consider an experiment where $F_{\kappa}$ in $\Pi^{\prime}$ is replaced by a truly random oracle $\mathcal{R}(\cdot)$. Since $x$ has not been queried, $\mathcal{R}(x)$ is a truly random element that was not queried before, and we may think of it as chosen before the initialization of $\Pi^{\prime}$. No n.u. p.p.t adversary $\mathcal{A}_{\dagger}$ can distinguish between the $\Pi^{\prime}$ we constructed using $\mathcal{R}(\cdot)$ and the $\Pi^{\prime}$ construction that uses the pseudo-random permutation $F_{\kappa}$ by more than a negligible advantage. We can prove this by contradiction. Suppose that there does exist a non-negligible function $\delta(\lambda)$ such that $\mathcal{A}_{\dagger}$ can attack $\Pi^{\prime}$ and find a false positive with probability $\epsilon + \delta(\lambda)$. We can run $A_{\dagger}$ on $\Pi^{\prime}$ where the oracle is replaced by an oracle that is either random or pseudo-random, and return $1$ if $\mathcal{A}_{\dagger}$ finds a false positive. This allows us to distinguish between $\mathcal{R}(\cdot)$ and $F_{\kappa}(\cdot)$ with probability $\geq \delta(\lambda)$. This contradicts the indistinguishability of pseudo-random permutations.

\textit{Correctness proof:} $\Pi^{\prime}$ is still a valid Bloom Filter as per Def.~\ref{def:bloom_filter}. $\Pi^{\prime}$'s completeness follows from the completeness of the original BF $\Pi$. From the soundness of $Pi$, we get that the probability of $x$ being a false positive in $\Pi^{\prime}$ is at most $\epsilon$. Therefore, the probability of $\mathcal{A}_{\dagger}$ winning the \textsc{Smash-Bloom} game is $\text{Pr}[\textsc{Smash-Bloom}(\mathcal{A}_{\dagger}, t, \lambda) = \top] \leq \varepsilon + \mathrm{negl}(\lambda)$.
\end{proof}

We can now prove our main result regarding the security of LSM stores.

\begin{theorem}\label{thm:secure_lsm}
    Let $\Lambda = (C_{r}, I_{r}, Q)$ be an LSM store using $m$ bits of memory for its Bloom Filters. If pseudo-random permutations exist, then there exists a negligible function $\mathrm{negl(\cdot)}$ such that for security parameter $\lambda$, there exists an LSM engine $\Lambda^{\prime}$ that is $(n, t, \epsilon + \mathrm{negl}(\lambda))$-secure against Bloom Filter targeting adversaries $\mathbb{A}_{\text{BF}}$ using ${m}^{'} = m + \lambda$ bits of memory for its Bloom Filters.
\end{theorem}

\begin{proof}

We first demonstrate a secure construction. We then prove its security and correctness.\vspace{0.2em}

\noindent \textit{Construction}: Choose a key $\kappa \in \{0, 1\}^{\lambda}$ for a pseudorandom permutation $F_{\kappa}$. Let $\Lambda^{\prime} = (C_{r}, I^{\prime}_{r}, Q^{\prime})$ where

\begin{enumerate}
    \item $I^{\prime}_{r}((k, v), \sigma) = C_{r}((F_{\kappa}(k), v), \sigma)$
    \item $Q^{\prime}(x, \sigma) = Q(F_{\kappa}(x), \sigma)$.
\end{enumerate}

\noindent The only additional memory required is for storing $\kappa$ which is $\lambda$ bits long.

\textit{Security Proof}: We prove this by contradiction. Suppose there does exist an n.u. p.p.t. adversary $\mathcal{A}$ that can with the \textsc{Smash-Lsm} game with probability greater than $\varepsilon + \textrm{negl}(\lambda)$. Then by definition of oracle $\mathscr{O}_{Q}$, there exists a Bloom Filter $\Pi^{\prime}$ with the construction from Thm~\ref{thm:secure_bloom_filter} for which $\mathcal{A}$ can generate false positives with probability higher than $\varepsilon + \mathrm{negl}(\lambda)$. So $\mathcal{A}$ can also win the \textsc{Smash-Bloom} game with a probability higher than $\varepsilon + \mathrm{negl}(\lambda)$. This is a contradiction as Thm~\ref{thm:secure_bloom_filter} proves that no n.u p.p.t adversary can win the \textsc{Smash-Bloom} game against Bloom Filter constructions of type $\Pi^{\prime}$ with probability higher than $\varepsilon + \mathrm{negl}(\lambda)$.

\textit{Correctness Proof}: The correctness of our construction follows from the fact that $F_{\kappa}$ is a bijection so it does not affect the correctness of the internal Bloom Filters, fence pointers, or the algorithms called on the LSM tree. 
\end{proof}

\noindent Since only the keys are permuted, not the values stored in the LSM store, we only need to do a forward permutation $F_{\kappa}$ but not an inverse permutation $F_{\kappa}^{-1}$ for a point query.

\section{Related Work} 

We discuss two areas of research relevant to our work: adversarial correctness of probabilistic data structures, and LSM store benchmarking and optimization.

\subsubsection*{Adversarial Data Structures.} There have been many recent efforts to rigorously define a game-based~\cite{naor_eylon,naor_oved,hayder_2024,clayton_2019} and simulator-based~\cite{filic_2022,filic_2024,hayder_2025} adversarial model for probabilistic data structures such as the Bloom Filter and the Learned Bloom Filter. In particular, the work of Naor et. al.~\cite{naor_eylon,naor_oved} was the first to show provably secure constructions for the Bloom Filter. There is also prior work showing feasible attacks on probabilistic data structures including the Bloom Filter~\cite{gerbet_2015} and the Learned Bloom Filter~\cite{reviriego_2021}. 
Our work builds upon these insights by applying them to LSM stores and designing countermeasures tailored to real-world storage systems. 

\subsubsection*{LSM Store Benchmarking / Optimization.} KVBench~\cite{zhu_et_al_2024} is one of many works that focus on benchmarking workloads for LSM stores. Prior work on LSM store optimizations includes ~\cite{monkey,thakkar_2024,spooky,huynh_et_al_2022}. Monkey~\cite{monkey} focuses on optimizing the memory budget allocation of the Bloom Filters used in LSM stores for better query performance. Huynch et. al.~\cite{huynh_et_al_2022} use Large Language Models (LLMs) to tune the design knobs of LSM stores. All prior work discussed primarily focuses on benchmarking and improving performance under non-adversarial workloads. To the best of our knowledge, this is the first paper to rigorously define an adversarial model for LSM stores and propose concrete, provably secure LSM store constructions.

\section{Open Problems}\label{sec:open_problems}

We leave the community with open problems in four categories that emerge from our work. 

\subsection{Adversarial Targets.} The experiments in our work focus on adversaries that target Bloom Filters in LSM stores via well-chosen insertions. We conduct our experiments on, LevelDB and RocksDB, both of which use a merge policy called leveling~\cite{monkey,huynh_et_al_2022}. There is a vast universe of adversarial targets we leave as open problems that emerge from our work. Empirical evidence showing performance degradation from these targets would help greatly in the design of future LSM stores.

\subsubsection*{Deleted Insertions} A direct follow-up for our attacks comes from the observation that Bloom Filters in LSM trees cannot perform deletions. This means that if an adversary $\mathcal{A}$ sabotages Bloom Filters via well-chosen insertions and then deletes the keys it inserted, the false positive rate of the Bloom Filters remains high unless the Bloom Filter is reconstructed from existing keys.

\subsubsection*{Range Queries} Range queries~\cite{monkey,shubham_subhadeep_2024} look for a larger number of keys in multiple levels of an LSM tree compared to point queries. Therefore, they are potentially more vulnerable to adversarial workloads. Adversary $\mathcal{A}$ can insert data in a manner that forces range queries to access an excessive number of runs, increasing read latency. Studying adversarial range query complexity and designing more efficient prefetching and merging strategies could help mitigate these attacks.

\subsubsection*{Merge Policy} There are different merge policies within LSM stores including leveling and tiering~\cite{monkey}. Prior work~\cite{huynh_et_al_2022} has shown that leveling is more robust to uncertain (but not explicitly adversarial) workloads than tiering. We conducted our evaluations on LevelDB and RocksDB, both of which use leveling. We might potentially see higher performance degradation on LSM stores that rely on tiering.

\subsection{General Adversary-resilience} In our work, we have shown a construction that provides adversary-resilience in an LSM store against $\mathbb{A}_{\text{BF}}$, the set of computationally bound Bloom Filter targeting adversaries. We leave answering the following follow-up questions as open problems:

\begin{enumerate}
    \item Does our construction provide adversary resilience against other classes of adversaries?
    \item Is there a construction, perhaps using our simulator-based model ($Appendix~\ref{sec:simulator_based_adversarial_model}$) guaranteeing reasonable adversary resilience against all computationally bound adversaries?
    \item Is a construction possible (either only for $\mathbb{A}_{\text{BF}}$ or for the general set $\mathbb{A}$) that provides adversary-resilience if the adversary is computationally-unbounded?
\end{enumerate}

\noindent Helpful starting points for this work include~\cite{naor_eylon,naor_oved} who provide an adversary-resilient construction of the Bloom Filter against computationally unbounded adversaries, and the simulator-based security constructions of~\cite{filic_2022,filic_2024}.

\subsection{Learned Bloom Filters.} We conjecture that replacing Bloom Filters used by an LSM store with Learned Bloom Filters~\cite{hayder_2024} with a better false positive rate may lead to better performance. Adversary resilience for Learned Bloom Filters can be solved using the construction of~\cite{hayder_2024}. We leave experimental validation of this as an open problem. Similar research for Adaptive Bloom Filters (or Learned Adaptive Bloom Filters~\cite{dai_shrivastava_2020}) will also be interesting. A helpful starting point for this work is~\cite{hayder_2024} who provide an adversary-resilient construction of the Learned Bloom Filter, called the Downtown Bodega Filter, in the same game-based setting as~\cite{naor_eylon}.

\subsection{Existence of an Ideal-World Simulator} In our simulator-based adversarial model (Appendix~\ref{sec:simulator_based_adversarial_model}), we have not proved the existence of an ideal-world simulator. We leave proving the existence of and providing a construction for an ideal-world simulator for LSM stores as an open problem. The constructions of an ideal-world simulator for probabilistic data structures by~\cite{filic_2022,filic_2024} may be a good starting point. However, their constructions require two properties in a data structure: function decomposability and reinsertion variance. Informally, reinsertion invariance requires the internal state of a probabilistic data structure to remain unchanged when the same key is reinserted. This does not necessarily apply to LSM stores.

\section{Conclusion}

In this paper, we investigate the performance of LSM stores under adversarial workloads. Our analysis shows that adversaries can significantly increase zero-result lookup latencies. We introduce a lightweight, provably secure mitigation strategy based on keyed pseudorandom permutations. Our implementation in LevelDB and RocksDB shows that this approach effectively reduces adversarial impact while maintaining overall system performance. Our work demonstrates the importance of adversarial resilience in storage systems. We introduce the community to several important open problems calling for both theoretical and experimental research into broader classes of attacks and secure designs for LSM stores.


\bibliographystyle{ACM-Reference-Format}
\bibliography{sample}

\appendix

\section{Simulator-based Model}\label{sec:simulator_based_adversarial_model}

In this section, we define a simulator-based~\cite{lindell_2016} adversarial model for LSM stores inspired by the simulator-based model of~\cite{filic_2022} for Bloom Filters. The game-based adversarial model of Section~\ref{sec:adversarial_model} is restricted to adversaries that target an LSM store's Bloom Filters. The simulator-based adversarial model, on the other hand, applies to any class of adversaries. We reuse the notation defined in Section~\ref{sec:preliminaries}.\vspace{0.2em}

\noindent \textbf{Adversarial Setting.} Let $\Lambda = (C_{r}, I_{r}, Q)$ be an LSM store, with public parameters $\Phi$. Let $\mathbb{A}$ be any given set of adversaries playing Game~\ref{game:real-or-ideal} that are bound by time $t_{\mathcal{\mathbb{A}}}$ and are allowed at most $\psi{i}, \psi{q}, \psi_{r}$ queries to oracles $\mathscr{O}_{I},\mathscr{O}_{Q},\mathscr{O}_{R}$ respectively. Let $\mathbb{D}$ be the set of all distinguishers bound by time $t_{\mathbb{D}}$. To establish any results regarding the behavior of LSM stores in the presence of adversaries, we need to compare it to what the behavior of an LSM store without the presence of an adversary is expected to be. We call any simulator that provides a view of the LSM store without adversarial interference an \textit{ideal-world} simulator. Let $\mathcal{S}$ be an ideal-world simulator bound by time $t_{\mathcal{S}}$ that provides a \emph{non-adversarially-influenced view} of $\Lambda$ to adversaries in set $\mathbb{A}$. For simpler notation, we denote $\Psi = (\psi_{i}, \psi_{q}, \psi_{r})$ and $T = (t_{\mathbb{A}}, t_{\mathcal{S}}, t_{\mathbb{D}})$.

\begin{definition}[Adversary Resilient LSM store]\,\newline
We say that $\Lambda$ is an $(\Psi, T, \varepsilon)$-resilient LSM store if for all adversaries $\mathcal{A} \in \mathbb{A}$ and for all distinguishers $\mathscr{D} \in \mathbb{D}$ in Game~\ref{game:real-or-ideal}, it holds that: 
\begin{align*}
	|\text{Pr}[{\text{Real}(\mathcal{A}, \mathscr{D}){=}1}]{-}\text{Pr}[{\text{Ideal}(\mathcal{A}, \mathscr{D}, \mathcal{S}){=}1}]| {\le} \varepsilon.
\end{align*}

\noindent where the constructions $\Lambda, \mathbb{A}, \mathbb{D}, \Psi$, and $T$ are defined in the Adversarial Setting section above. The probabilities are taken over the random coins used by $I_{r}$ and $_C{r}$ within $\mathscr{O}_{I}$ and $\mathscr{O}_{C}$, as well as any random coins used by $\mathcal{A}, \mathcal{S}$, or $\mathscr{D}$.
\end{definition}

\begin{game}\label{game:real-or-ideal}
We have an adversary $\mathcal{A}$, a simulator $\mathcal{S}$, a distinguisher $\mathscr{D}$, and public parameters $\Phi$. 

\begin{pchstack}[boxed]
    \begin{pcvstack}[]
    \procedure[linenumbering,headlinecmd={\vspace{.1em}\hrule\vspace{.3em}}]{Real-Or-Ideal($\mathcal{A}, \mathcal{S}, \mathscr{D}, \Phi$)}{
    d \sampledfrom \{0, 1\} \\ 
    \pcif d = 0 \t \pcmycomment{Real} \\
    \ \ \sigma \sampledfrom C_{r}(\Phi)\\
    \ \ y \sampledfrom \mathcal{A}^{\mathscr{O}_{I}, \mathscr{O}_{Q}, \mathscr{O}_{R}} \\ 
    \pcelse \t \pcmycomment{Ideal} \\
    \ \ y \sampledfrom \mathcal{S}(\mathcal{A}, \Phi) \\
    \pcreturn d^{\prime} \sampledfrom \mathscr{D}(y) 
    }
    \end{pcvstack}
    \hspace{5em}
    \begin{pcvstack}
        \procedure[headlinecmd={\vspace{.1em}\hrule\vspace{.3em}}]{Oracle $\mathscr{O}_{I}(k, v)$}{%
            \sigma \sampledfrom I_{r}(k, v, \sigma)
        }
        \vspace{1.2em}
        \procedure[headlinecmd={\vspace{.1em}\hrule\vspace{.3em}}]{Oracle $\mathscr{O}_{Q}(k)$}{%
            \pcreturn Q(k, \sigma)
        }
        \vspace{1.2em}
        \procedure[space=1em,headlinecmd={\vspace{.1em}\hrule\vspace{.3em}}]{Oracle $\mathscr{O}_{R}()$}{%
            \pcreturn \sigma
        }
    \end{pcvstack}
\end{pchstack}

\end{game}

\end{document}